\documentclass[12pt]{article}

\usepackage{amsmath,amssymb,amsthm,amscd,amsfonts}
\usepackage{fullpage}
\usepackage{graphicx}
\usepackage{stmaryrd}
\usepackage{float}
\usepackage{hyperref}
\DeclareMathOperator{\lcm}{lcm}
\usepackage{algorithm}							
\usepackage[noend]{algpseudocode}					
\usepackage{caption}
\usepackage[margin=1in]{geometry}
\usepackage{url}
\def\divides{{\, | \, }}

\def\Enn{{\mathbb{Z}}}

\title{Periodicity in Rectangular Arrays}

\author{Guilhem Gamard\\
LIRMM \\
CNRS, Univ. Montpellier \\
UMR 5506, CC 477 \\
161 rue Ada \\
34095 Montpellier Cedex 5 \\
France \\
\href{mailto:guilhem.gamard@lirmm.fr}{\tt guilhem.gamard@lirmm.fr} \\
\and
Gwena\"el Richomme \\
LIRMM\\
CNRS, Univ. Montpellier\\ 
UMR 5506, CC 477 \\ 
161 rue Ada \\ 
34095 Montpellier Cedex 5 \\ 
France \\
and\\
Univ. Paul-Val\'ery Montpellier 3 \\ Route de Mende \\
34199 Montpellier Cedex 5 \\ France \\
\href{mailto:gwenael.richomme@lirmm.fr}{\tt gwenael.richomme@lirmm.fr} \\
\and
Jeffrey Shallit and Taylor J. Smith\\
School of Computer Science \\
University of Waterloo \\
Waterloo, ON  N2L 3G1 \\
Canada \\
\href{mailto:shallit@cs.uwaterloo.ca}{\tt shallit@cs.uwaterloo.ca} \\
\href{mailto:tj2smith@uwaterloo.ca}{\tt tj2smith@uwaterloo.ca} 
}

\begin{document}

\maketitle

\theoremstyle{plain}
\newtheorem{theorem}{Theorem}
\newtheorem{corollary}[theorem]{Corollary}
\newtheorem{lemma}[theorem]{Lemma}
\newtheorem{proposition}[theorem]{Proposition}

\theoremstyle{definition}
\newtheorem{definition}[theorem]{Definition}
\newtheorem{example}[theorem]{Example}
\newtheorem{conjecture}[theorem]{Conjecture}

\theoremstyle{remark}
\newtheorem{remark}[theorem]{Remark}
\newtheorem{fact}[theorem]{Fact}

\begin{abstract}
We discuss several two-dimensional generalizations of
the familiar Lyndon-Sch\"utzenberger periodicity theorem for words.  We consider
the notion of primitive array (as one that cannot be expressed as the
repetition of smaller arrays).  We count the number of $m \times n$
arrays that are primitive.  Finally, we show that one can test primitivity
and compute the primitive root of an array in linear time.

\medskip


\medskip

\end{abstract}

\section{Introduction}

Let $\Sigma$ be a finite alphabet.  One very general version of the
famous Lyndon-Sch\"utzenberger
theorem \cite{Lyndon&Schutzenberger:1962} can be stated as follows:

\begin{theorem}
Let $x, y \in \Sigma^+$.  Then the following five conditions are
equivalent:

(1) $xy = yx$;

(2) There exist $z \in \Sigma^+$ and integers $k, \ell > 0$ such that
$x = z^k$ and $y = z^\ell$;

(3) There exist integers $i, j > 0$ such that $x^i = y^j$;

(4) There exist integers $r, s > 0$ such that $x^r y^s = y^s x^r$;

(5) $x\{ x,y\}^* \ \cap \ y\{x,y\}^* \not= \emptyset$.
\label{ls-thm}
\end{theorem}

\begin{proof}
For a proof of the equivalence of (1), (2), and (3), see, for
example \cite[Theorem 2.3.3]{Shallit:2009}.   

Condition (5) is essentially the ``defect theorem''; see, for
example, \cite[Cor.~1.2.6]{Lothaire:1983}.

For completeness, we now demonstrate the equivalence
of (4) and (5) to each other and to conditions
(1)--(3):

\medskip

\noindent (3) $\implies$ (4):   If $x^i = y^j$, then we immediately have 
$x^r y^s = y^s x^r$ with $r = i$ and $s = j$.

\medskip

\noindent (4) $\implies$ (5):  Let $z = x^r y^s$.  Then by (4) we have
$z = y^s x^r$.    So $z = x x^{r-1} y^s$ and $z = y y^{s-1} x^r$.
Thus $z \in x \{x,y\}^*$ and $z \in y \{x,y\}^*$.  So
$x\{ x,y\}^* \ \cap \ y\{x,y\}^* \not= \emptyset$.

\medskip

\noindent (5) $\implies$ (1):
By induction on the length of $|xy|$.  The base
case is $|xy| = 2$.  More generally, if $|x|=|y|$ then clearly
(5) implies $x=y$ and so (1) holds.  Otherwise without loss of
generality $|x| < |y|$.  Suppose $z \in x\{ x,y\}^*$ and
$z \in y\{x,y\}^*$.  Then $x$ is a proper prefix of $y$, so write
$y = xw$ for a nonempty word $w$.  Then $z$ has prefix $xx$
and also prefix $xw$.  Thus $x^{-1} z \in x \{ x, w \} ^*$ and
$x^{-1} z \in w \{ x, w \}^*$, where by $x^{-1} z$ we mean remove
the prefix $x$ from $z$.  So $x \{ x, w \} ^* \ \cap \ w \{ x, w \}^*
\not= \emptyset$, so by induction (1) holds for $x$ and $w$, so 
$xw = wx$.  Then $yx = (xw)x = x(wx) = xy$.
\end{proof}

A nonempty word $z$ is {\it primitive\/} if it cannot be written in the
form $z = w^e$ for a word $w$ and an integer $e \geq 2$.  We will need the
following fact (e.g., \cite[Prop.~1.3.1]{Lothaire:1983} or
\cite[Thm.~2.3.4]{Shallit:2009}):

\begin{fact}
Given a nonempty word $x$, the shortest word $z$ such that
$x = z^i$ for some integer $i \geq 1$ is primitive.   It is called
the {\it primitive root} of $x$, and is unique.
\end{fact}

In this paper we consider generalizations of the
Lyndon-Sch\"utzenberger theorem and the notion of primitivity to
two-dimensional rectangular arrays (sometimes called {\it pictures} in
the literature).  For more about basic operations on these arrays, see,
for example, \cite{Giammarresi&Restivo:1997}.

\section{Rectangular arrays}

By $\Sigma^{m\times n}$ we mean the set of all $m \times n$ rectangular
arrays $A$ of elements chosen from the alphabet $\Sigma$.  Our arrays
are indexed starting at position $0$, so that $A[0,0]$ is the element
in the upper left corner of the array $A$.  We use the notation
$A[i..j,k..\ell]$ to denote the rectangular subarray with rows
$i$ through $j$ and columns $k$ through $\ell$.  If $A \in \Sigma^{m\times n}$,
then $|A| = mn$ is the number of entries in $A$.

We also generalize
the notion of powers as follows.  If $A \in \Sigma^{m \times n}$
then by $A^{p \times q}$ we mean the array 
constructed by repeating $A$ $pq$ times, in $p$ rows and $q$ columns.
More formally $A^{p \times q}$ is the $pm \times qn$ array $B$
satisfying $B[i,j] = A[i \bmod m,j \bmod n]$
for $0 \leq i < pm$ and $0 \leq j < qn$.
For example, if 
$$ A = \left[\begin{array}{ccc} 
	{\tt a} & {\tt b} & {\tt c} \\
	{\tt d} & {\tt e} & {\tt f}
	\end{array}\right] ,$$
then 
$$ A^{2\times 3} = \left[\begin{array}{ccccccccc} 
	{\tt a} & {\tt b} & {\tt c} & {\tt a} & {\tt b} & {\tt c} & {\tt a} & {\tt b} & {\tt c}\\
	{\tt d} & {\tt e} & {\tt f} & {\tt d} & {\tt e} & {\tt f} & {\tt d} & {\tt e} & {\tt f}\\
	{\tt a} & {\tt b} & {\tt c} & {\tt a} & {\tt b} & {\tt c} & {\tt a} & {\tt b} & {\tt c}\\
	{\tt d} & {\tt e} & {\tt f} & {\tt d} & {\tt e} & {\tt f} & {\tt d} & {\tt e} & {\tt f}
	\end{array}
	\right] .$$

We can also generalize the notation of concatenation of arrays, but now
there are two annoyances:  first, we need to decide if we are
concatenating horizontally or vertically, and second, to obtain a
rectangular array, we need to insist on a matching of dimensions.

If $A$ is an $m \times n_1$ array and $B$ is an $m \times n_2$ array,
then by $A \obar B$ we mean the $m \times (n_1 + n_2)$ array obtained
by placing $B$ to the right of $A$.

If $A$ is an $m_1 \times n$ array and $B$ is an $m_2 \times n$ array,
then by $A \ominus B$ we mean the $(m_1+m_2) \times n$ array obtained
by placing $B$ underneath $A$.

\section{Generalizing the Lyndon-Sch\"utzenberger theorem}

We now state our first generalization of the Lyndon-Sch\"utzenberger theorem
to two-dimensional arrays, which generalizes
claims (2), (3), and (4) of Theorem~\ref{ls-thm}.

\begin{theorem}
Let $A$ and $B$ be nonempty arrays.  Then the following three conditions
are equivalent:

(a) There exist
positive integers $p_1, p_2, q_1, q_2$ such that
$A^{p_1 \times q_1} = B^{p_2 \times q_2}$.

(b) There exist a nonempty array $C$
and positive integers $r_1, r_2, s_1, s_2$ such that
$A = C^{r_1 \times s_1}$ and $B = C^{r_2 \times s_2}$.

(c) There exist positive integers $t_1, t_2, u_1, u_2$ such that
$A^{t_1 \times u_1} \circ B^{t_2 \times u_2} =
B^{t_2 \times u_2} \circ A^{t_1 \times u_1}$
where $\circ$ can be either $\obar$ or $\ominus$.
\label{one}
\end{theorem}

\begin{proof}
\ \vphantom{a} \\
(a) $\implies$ (b). Let $A$ be an array in $\Sigma^{m_1 \times n_1}$
and $B$ be an array in $\Sigma^{m_2 \times n_2}$
such that $A^{p_1 \times q_1} = B^{p_2 \times q_2}$.
By dimensional considerations we have $m_1p_1 = m_2p_2$ and $n_1q_1 = n_2q_2$. 
Define $P = A^{p_1 \times 1}$ and $Q = B^{p_2 \times 1}$. 
We have $P^{1\times q_1} = Q^{1\times q_2}$.
Viewing $P$ and $Q$ as words over $\Sigma^{m_1p_1 \times 1}$ and considering horizontal concatenation, this can be written
$P^{q_1} = Q^{q_2}$. 
By Theorem~\ref{ls-thm} there exist a word
$R$ over $\Sigma^{m_1p_1 \times 1}$ and integers $s_1, s_2$ such that
$P = R^{1\times s_1}$ and $Q = R^{1\times s_2}$. 
Let $r$ denote the number of columns of $R$ and
let $S = A[0 \dots m_1-1, 0 \dots r-1]$ and
$T = B[0\dots m_2-1, 0 \dots r-1]$. 
Observe $A = S^{1 \times s_1}$ and $B = T^{1\times s_2}$.
Considering the $r$ first columns of $P$ and $Q$, 
we have $S^{p_1 \times 1} = T^{p_2 \times 1}$.
Viewing $S$ and $T$ as words over
$\Sigma^{1 \times r}$ and considering vertical concatenation, we can rewrite $S^{p_1} = T^{p_2}$. By
Theorem~\ref{ls-thm} again, there exist a word $C$ over
$\Sigma^{1 \times r}$ and integers $r_1, r_2$ such that $S = C^{r_1 \times 1}$
and $T = C^{r_2 \times 1}$. Therefore,
$A = C^{r_1 \times s_1}$ and $B = C^{r_2 \times s_2}$.

\bigskip

\noindent (b) $\implies$ (c). Without loss of generality, assume that
the concatenation operation is $\obar$.  Let us recall that
$A = C^{r_1 \times s_1}$ and $B = C^{r_2\times s_2}$. Take
$t_1 = r_2$ and $t_2 = r_1$ and $u_1 = s_2$ and $u_2 = s_1$. Then we
have
\begin{align*}
  A^{t_{1} \times u_{1}} \obar B^{t_2 \times u_{2}}	
  &=	C^{r_{1}t_{1} \times s_{1}u_{1}} \obar C^{r_{2}t_{2} \times s_{2}u_{2}} \\
  &=	C^{r_1 t_1 \times (s_1 u_1 + s_2 u_2)}  && \text{(Observe that $r_1t_1 = r_2t_2$)} \\
  &=	C^{r_2 t_2 \times s_2 u_2} \obar C^{r_1 t_1 \times s_1 u_1}  \\
  &=	B^{t_2 \times u_{2}} \obar A^{t_{1} \times u_{1}}.
\end{align*}

\bigskip

\noindent (c) $\implies$ (a). Without loss of generality,
assume that the concatenation operation is $\obar$. Assume the
existence of positive integers $t_1, t_2, u_1, u_2$ such that
\begin{equation*}
  A^{t_1 \times u_1} \obar B^{t_2 \times u_2} =
  B^{t_2 \times u_2} \obar A^{t_1 \times u_1}.
\end{equation*}
An immediate induction allows to prove that for all positive integers
$i$ and $j$,
\begin{equation}
  A^{t_1 \times iu_1} \obar B^{t_2 \times ju_2} =
  B^{t_2 \times ju_2} \obar A^{t_1 \times iu_1}.
  \label{eqproof3}
\end{equation}
Assume that $A$ is in
$\Sigma^{m_1 \times n_1}$ and $B$ is in $\Sigma^{m_2 \times n_2}$. For
$i = n_2 u_2$ and $j = n_1 u_1$, we get $i u_1 n_1 = j u_2
n_2$. Then, by considering the first $i u_1 n_1$ columns of the
array defined in~\eqref{eqproof3}, we get
$A^{t_1 \times iu_1} = B^{t_2 \times j u_2}$.
\end{proof}

Note that generalizing condition~(1) of Theorem~\ref{ls-thm} requires
considering arrays with the same number of rows or same number of
columns. Hence the next result is a direct consequence of the previous
theorem.

\begin{corollary}
Let $A, B$ be nonempty rectangular arrays.   Then

\noindent (a) if $A$ and $B$ have the same number
of rows, $A \obar B = B \obar A$ if and only there exist a nonempty
array $C$ and integers $e, f \geq 1$ such that
$A = C^{1 \times e}$ and $B = C^{1 \times f}$;

\noindent (b) if $A$ and $B$ have the same number
of columns , $A \ominus B = B \ominus A$ if and only there exist a nonempty 
array $C$ and integers $e, f \geq 1$ such that
$A = C^{e \times 1}$ and $B = C^{f \times 1}$.
\label{foure}
\end{corollary}

\section{Labeled plane figures}

We can generalize condition~(5) of Theorem~\ref{ls-thm}. We begin with
the following lemma. As in the case of Corollary~\ref{foure}, we need
conditions on the dimensions.

\begin{lemma}
Let $X$ and $Y$ be rectangular arrays having same number of rows or same numbers of columns. 
In the former case set $\circ =  \obar$. In the latter case set $\circ =  \ominus$. If 
\begin{equation}
X \circ W_1 \circ W_2 \circ \cdots \circ W_i = Y \circ Z_1 \circ Z_2 
\circ \cdots \circ Z_j
\label{eq1}
\end{equation}
holds, where $W_1, W_2, \ldots, W_i, Z_1, Z_2, \ldots, Z_j \in \{X, Y\}$
for $i, j \geq 0$, then $X$ and $Y$ are powers of a third array $T$.
\label{four}
\end{lemma}

\begin{proof}
Without loss of generality we can assume that $X$ and $Y$ have the same number $r$ of rows.
Then the lemma is just a rephrasing of part (5) $\implies$ (2) in Theorem~\ref{ls-thm}, 
considering $X$ and $Y$ as words over $\Sigma^{r\times 1}$.
\end{proof}

%

Now we can give our maximal generalization of (5) $\implies$ (3) in
Theorem~\ref{ls-thm}.  To do so, we need the concept of labeled plane
figure (also called ``labeled polyomino'').

A {\it labeled plane figure} is a finite union of labeled cells in the
plane lattice, that is, a map from a finite subset of
$\Enn \times \Enn$ to a finite alphabet $\Sigma$. A sample plane
figure is depicted in Figure~\ref{lpf}. Notice that such a figure does
not need to be connected or convex.

\begin{figure}[H]
\begin{center}
\includegraphics{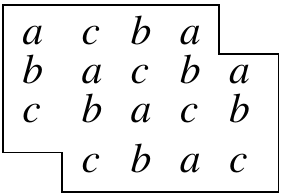}
\end{center}
\caption{A typical plane figure (from \cite{Huova:2009,Huova:2014})}
\label{lpf}
\end{figure}

Let $S$ denote a finite set of rectangular arrays.
A {\it tiling} of a labeled plane figure $F$ is an arrangement
of translates of the arrays in $S$ so that the label of every cell
of $F$ is covered by an identical entry of an element of $S$, 
and no cell of $F$ is covered by more than one entry of an element of $S$.
For example, Figure~\ref{lpf3} depicts a tiling of the labeled plane
figure in Figure~\ref{lpf} by the arrays
$[c \ b \ a]$, 
$\left[ \begin{array}{c} 
	a \\
	b \\
	c
	\end{array} \right]$,
and
$\left[ \begin{array}{cccc}
a & c & b & a \\
b & a & c & b \\
c & b & a & c
\end{array} \right]$.

\begin{figure}[H]
\begin{center}
\includegraphics{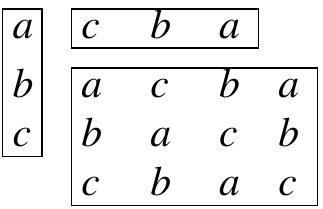}
\end{center}
\caption{Tiling of Figure~\ref{lpf}}
\label{lpf3}
\end{figure}

\begin{theorem}
Let $F$ be a labeled plane figure, and
suppose $F$ has two different tilings $U$ and $V$ by
two nonempty rectangular arrays $A$
and $B$.  Then both $A$ and $B$ are powers of a third array
$C$.
\label{plane}
\end{theorem}

\begin{proof}
Assume that $F$ has two different tilings by rectangular arrays,
but $A$ and $B$ are not powers of a third array $C$.
Without loss of generality also assume that $F$
is the smallest such figure (with the fewest cells) and also that
$A$ and $B$ are arrays with the fewest total entries
that tile $F$, but are not powers of a third array.

Consider the leftmost cell $L$ in the top row of $F$.  If this cell is
covered by the same array, in the same orientation, in both tilings $U$ and $V$,
remove the array from $U$ and $V$, obtaining a smaller plane figure
$F'$ with the same property.   This is a contradiction, since $F$ was
assumed minimal. So $F$ must have a different array in $U$ and $V$ at
this cell. Assume $U$ has $A$ in its tiling and $V$ has $B$.

Without loss of generality, assume that the number of rows of $A$
is equal to or larger than $r$, the number of rows of $B$.
Truncate $A$ at the
first $r$ rows and call it $A'$.  Consider the topmost row of $F$.  Since it
is topmost and contains $L$ at the left, there must be nothing above $L$.
Hence the topmost row of $F$ must be tiled with the topmost rows of $A$ and $B$
from left to right, aligned at this topmost row, until either the right
end of the figure or an unlabeled cell is reached.
Restricting our attention to the $r$ rows
underneath this topmost row, we get a rectangular tiling of these $r$ rows
by arrays $A'$ and $B$ in both cases, but the tiling of $U$ begins with $A'$
and the tiling of $V$ begins with $B$.

Now apply Lemma~\ref{four} to these $r$ rows (with $\circ = \obar$).  We get that $A'$ and $B$
are both expressible as powers of some third array $T$.
Then we can write $A$ as a
concatenation of some copies of $T$ and the remaining rows of $A$ (call the
remaining rows $C$). Thus we get two tilings of $F$ in terms of $T$ and $C$.
Since $A$ and $B$ were assumed to be the smallest nonempty tiles that could
tile $F$, and $|T| \leq |B|$ and $|C| < |A|$,
the only remaining possibility is
that $T = B$ and $C$ is empty.  But then $A = A'$ and so both $A$ and $B$ are
expressible as powers of $T$.
\end{proof}

\begin{remark}
The papers \cite{Moczurad&Moczurad:2004,Moczurad:2007} claim 
a proof of Theorem~\ref{plane}, but
the partial proof provided is incorrect in some details and
missing others.
\end{remark}

\begin{remark}
As shown by Huova \cite{Huova:2009,Huova:2014}, Theorem~\ref{plane} 
is not true for three rectangular arrays.  For example, the plane figure
in Figure~\ref{lpf} has the tiling in Figure~\ref{lpf3} and also another one.
\end{remark}

\section{Primitive arrays}

In analogy with the case of ordinary words, we can define the
notion of primitive array.  An array $M$ is said to be
{\it primitive\/} if the equation $M = A^{p \times q}$ for $p, q > 0$ implies that
$p = q = 1$.  For example, the array
$$ \left[ \begin{array}{cc}
	1 & 2 \\
	2 & 1 
	\end{array} \right]$$
is primitive, but 
$$ \left[ \begin{array}{cc}
	1 & 1 \\
	1 & 1
	\end{array}\right] \text{and}
\left[ \begin{array}{cc}
	1 & 2 \\
	1 & 2
	\end{array}\right]	
	$$ 
are not, as they can be written in the form $[1]^{2 \times 2}$ or $[1\; 2]^{2 \times 1}$ respectively.

As a consequence of Theorem~\ref{one} we get another proof of Lemma 3.3 in \cite{Gamard&Richomme:2015}.
\begin{corollary}
Let $A$ be a nonempty array.  Then there exist a unique primitive
array $C$ and positive integers $i, j$ such that $A = C^{i \times j}$.
\label{prim}
\end{corollary}

\begin{proof}
Choose $i$ as large as possible such that there exist an
integer $k$ and an array $D$ such that $A = D^{i \times k}$.  
Now choose $j$ as large as possible such that there
exists an integer $j$ and an array $C$ such that $A = C^{i \times j}$.
We claim that $C$ is primitive.  For if not, then there exists
an array $B$ such that $C = B^{i' \times j'}$ for positive integers
$i, j$, not both $1$.  Then $A = C^{i \times j} = B^{ii' \times  jj'}$,
contradicting
either the maximality of $i$ or the maximality of $j$.

For uniqueness, assume $A = C^{i_1 \times j_1} = D^{i_2 \times  j_2}$ where
$C$ and $D$ are both primitive.  Then by Theorem~\ref{one}
there exists an array $E$ such that $C = E^{p_1 \times  q_1}$ and
$D = E^{p_2 \times q_2}$.  Since $C$ and $D$ are primitive, we must
have $p_1 = q_1 = 1$ and $p_2 = q_2 = 1$.  Hence $C = D$.
\end{proof}

\begin{remark}
In contrast, as Bacquey \cite{Bacquey:2015} has recently shown, 
two-dimensional biperiodic {\it infinite\/} arrays can have two 
distinct primitive roots.
\end{remark}

\section{Counting the number of primitive arrays}

There is a well-known formula for the number of primitive
words of length $n$ over a $k$-letter alphabet (see \textit{e.g.} \cite[p. 9]{Lothaire:1983}):
\begin{equation}
\psi_k(n) = \sum_{d | n} \mu(d) k^{n/d},
\label{psi}
\end{equation}
where $\mu$ is the well-known M\"obius function, defined as follows:
$$
\mu(n) = \begin{cases}
	(-1)^t, & \text{if $n$ is squarefree and the product of
		$t$ distinct primes;} \\
	0, & \text{if $n$ is divisible by a square $> 1$}.
	\end{cases}
$$
We recall the following 
well-known property of the sum of the M\"{o}bius function $\mu (d)$
(see, e.g., \cite[Thm.~263]{Hardy&Wright:2008}):
\begin{lemma}\label{lem:musum} 
\begin{equation*}
\sum_{d | n} \mu(d) = 
\begin{cases}
1,	&\text{if } n = 1; \\
0,	&\text{if } n > 1.
\end{cases}
\end{equation*}
\label{moebius}
\end{lemma}
In this section we generalize Eq.~(\ref{psi}) to two-dimensional
primitive arrays:

\begin{theorem}
There are
$$  \psi_k(m,n) = \sum_{d_1 | m} \ \sum_{d_2 | n} \mu(d_1) \mu(d_2) k^{mn/(d_1 d_2)} $$
primitive arrays of dimension $m \times n$ over a $k$-letter
alphabet.
\end{theorem}

\begin{proof}

We will use Lemma~\ref{moebius} to prove our generalized formula, which
we obtain via M\"{o}bius inversion.

Define $g(m,n) := k^{mn}$; this counts the number of $m \times n$ arrays over
a $k$-letter alphabet.    Each such array has, by Corollary~\ref{prim},
a unique primitive root of dimension $d_1 \times d_2$, where evidently
$d_1 \ | \ m$ and $d_2 \ | \ n$.    So
$g(m,n) = \sum_{\substack{d_{1} | m \\ d_{2} | n}} \psi_k(d_{1}, d_{2})$.
Then
\begin{align*}
\sum_{\substack{d_{1} | m \\ d_{2} | n}} \mu(d_{1}) \mu(d_{2}) \ g \left( \frac{m}{d_{1}}, \frac{n}{d_{2}} \right) &= \sum_{d_{1} | m} \mu(d_{1}) \sum_{d_{2} | n} \mu(d_{2}) \ g \left( \frac{m}{d_{1}}, \frac{n}{d_{2}} \right) \\
&= \sum_{d_{1} | m} \mu(d_{1}) \sum_{d_{2} | n} \mu(d_{2}) \sum_{\substack{c_{1} | m / d_{1} \\ c_{2} | n / d_{2}}} \psi_k(c_{1}, c_{2}) \\
&= \sum_{c_{1}d_{1} | m} \mu(d_{1}) \sum_{c_{2}d_{2} | n} \mu(d_{2}) \ \psi_k(c_{1}, c_{2}) \\
&= \sum_{c_{1} | m} \sum_{c_{2} | n} \ \psi_k(c_{1}, c_{2}) \sum_{\substack{d_{1} | m / c_{1} \\ d_{2} | n / c_{2}}} \mu(d_{1}) \mu(d_{2}).
\end{align*}
Let $r = m / c_{1}$ and $s = n / c_{2}$. By Lemma~\ref{lem:musum}, the
last sum in the above expression is 1 if  $r = 1$ and $s = 1$; that is,
if $c_{1} = m$ and $c_{2} = n$. Otherwise, the last sum is 0. Thus, the
sum reduces to $\psi_k(m, n)$ as required.
\end{proof}

The following table gives the first few values of the function
$\psi_2(m,n)$:
\begin{table}[H]
\begin{center}
\begin{tabular}{c|cccccccc}
& 1 & 2 & 3 & 4 & 5 & 6 & 7  \\
\hline
1 & 2&              2&              6&             12&             30&             54& 126&            \\
2 & 2&             10&             54&            228&            990&           3966&          16254&          \\
3 & 6&             54&            498&           4020&          32730&         261522&        2097018&        \\
4 & 12&            228&           4020&          65040&        1047540&       16768860&      268419060&      \\
5 & 30&            990&          32730&        1047540&       33554370&     1073708010&    34359738210&    \\
6 & 54&           3966&         261522&       16768860&     1073708010&    68718945018&  4398044397642&  \\
7 & 126&          16254&        2097018&      268419060&    34359738210&  4398044397642&562949953421058&
\end{tabular}
\end{center}
\end{table}

\begin{remark}
As a curiosity, we note that $\psi_2(2,n)$ also counts the number of
{\it pedal triangles} with period exactly $n$.  See
\cite{Valyi:1903,Kingston&Synge:1988}.
\end{remark}

\section{Checking primitivity in linear time}

In this section we give an algorithm to test primitivity of
two-dimensional arrays.  We start with a useful lemma.

\begin{lemma}
Let $A$ be an $m \times n$ array.  Let the primitive root
of row $i$ of $A$ be $r_i$ and the primitive root of column
$j$ of $A$ be $c_j$.  Then the primitive root of $A$ has dimension
$p \times q$, where $q = \lcm(|r_0|, |r_1|, \ldots, |r_{m-1}|)$
and $p = \lcm(|c_0|, |c_1|, \ldots, |c_{n-1}|)$.
\label{three}
\end{lemma}

\begin{proof}
Let $P$ be the primitive root of the array $A$, of dimension
$m' \times n'$.  Then the row $A[i,0..n-1]$ is periodic with period $n'$.
But since the primitive root of $A[i,0..n-1]$ is of length
$r_i$, we know that $|r_i|$ divides $n'$.    It follows
that $q \divides n'$, where 
$q = \lcm(|r_0|, |r_1|, \ldots, |r_{m-1}|)$.
Now suppose $n' \not= q$.  Then since $q \divides n'$ we must
have $n'/q > 1$.  Define $Q := P[0..m'-1,0..q-1]$.  Then
$Q^{1 \times (n'/q)} = P$, contradicting our hypothesis that
$P$ is primitive.   It follows that $n' = q$, as claimed.

Applying the same argument to the columns proves the claim about $p$.
\end{proof}

Now we state the main result of this section.

\begin{theorem}
We can check primitivity of an $m \times n$
array and compute the primitive root in $O(mn)$ time, for
fixed alphabet size.
\label{primcheck}
\end{theorem}

\begin{proof}
As is well known,
a word $u$ is primitive if and only if $u$ is not an interior factor of its
square $uu$ \cite{ChoffrutKarhumaki1997Combinatorics}; that is, $u$ is
not a factor of the word $u_{F}u_{L}$, where $u_{F}$ is $u$ with the
first letter removed and $u_{L}$ is $u$ with the last letter removed.
We can test whether $u$ is a factor of $u_{F}u_{L}$ using a linear-time
string matching algorithm, such as the Knuth-Morris-Pratt algorithm
\cite{KnuthMorrisPratt1977FastPatternMatching}. If the algorithm
returns no match, then $u$ is indeed primitive.  Furthermore,
if $u$ is not primitive, the length of its primitive root is given
by the index (starting with position $1$) of the first match of
$u$ in $u_{F}u_{L}$.  
We assume that there
exists an algorithm \textsc{1DPrimitiveRoot} to obtain the primitive
root of a given word in this manner.

We use Lemma~\ref{three} as our basis for the following
algorithm to compute the primitive root of a rectangular array. This
algorithm takes as input an array $A$ of dimension $m \times n$ and
produces as output the primitive root $C$ of $A$ and its dimensions.

\begin{algorithm}
\caption{Computing the primitive root of $A$}\label{alg:computeprimitiveroot}
\begin{algorithmic}[1]
\Procedure{2DPrimitiveRoot}{$A,m,n$}
	\For{$0 \leq i < m$}								\Comment{compute primitive root of each row}
		\State $r_{i} \gets \textsc{1DPrimitiveRoot}(A[i, \ 0..n-1])$
	\EndFor
	\State $q \gets \lcm(|r_{0}|, |r_{1}|, \dots, |r_{m-1}|)$		\Comment{compute lcm of lengths of primitive roots of rows}
	\For{$0 \leq j < n$}
	\Comment{compute primitive root of each column}
		\State $c_{j} \gets \textsc{1DPrimitiveRoot}(A[0..m-1, j])$
	\EndFor
	\State $p \gets \lcm(|c_0|, |c_1|, \dots, |c_{n-1}|)$ \Comment{compute lcm of lengths of primitive roots of columns}
	\For{$0 \leq i < p$}
		\For{$0 \leq j < q$}
			\State $C[i,j] \gets A[i,j]$
			\EndFor
		\EndFor
	\State \textbf{return} $(C, p, q)$
\EndProcedure
\end{algorithmic}
\end{algorithm}

The correctness follows immediately from Lemma~\ref{three}, and
the running time is evidently $O(mn)$.
\end{proof}

\begin{remark}
The literature features a good deal of previous work on pattern
matching in two-dimensional arrays. The problem of finding every
occurrence of a fixed rectangular pattern in a rectangular array was
first solved independently by Bird \cite{Bird1977PatternMatching} and
by Baker \cite{Baker1978StringMatching}. Amir and Benson later
introduced the notion of two-dimensional periodicity in a series of
papers \cite{AmirBenson1992AlphabetIndependent,
AmirBenson19922DPeriodicity, AmirBenson19982DPeriodicity}. 
Mignosi, Restivo, and Silva \cite{Mignosi&Restivo&Silva:2003} considered
two-dimensional generalizations of the Fine-Wilf theorem.  A survey of
algorithms for two-dimensional pattern matching may be found in Chapter
12 of Crochemore and Rytter's text
\cite{CrochemoreRytter1994TextAlgorithms}.  
Marcus and Sokol \cite{Marcus&Sokol:2013} considered two-dimensional
Lyndon words.  Crochemore, Iliopoulos, and Korda
\cite{Crochemore&Iliopoulos&Korda:1998} and, more recently, Gamard and
Richomme \cite{Gamard&Richomme:2015}, considered
quasiperiodicity in two dimensions.  
However, with the exception of this latter paper, where Corollary~\ref{prim}
can be found, none of
this work is directly related to the problems we consider in this
paper.
\end{remark}

\begin{remark}
One might suspect that it is easy to reduce 2-dimensional primitivity
to 1-dimensional primitivity by considering the array $A$ as a 
1-dimensional word, and taking the elements in row-major
or column-major order.  However, the natural conjectures that
$A$ is primitive if and only if (a) either its corresponding
row-majorized or column-majorized
word is primitive, or (b) both its row-majorized or column-majorized
words are primitive, both fail.
For example,
assertion (a) fails because 
$$ \left[ \begin{array}{cc}
{\tt a} & {\tt a} \\
{\tt b} & {\tt b}
\end{array}\right]$$
is not primitive, while its row-majorized word ${\tt aabb}$ 
is primitive.
Assertion (b) fails because
$$ \left[ \begin{array}{ccc}
{\tt a} & {\tt b} & {\tt a} \\
{\tt b} & {\tt a} & {\tt b} 
\end{array} \right]$$
is 2-dimensional primitive, but its row-majorized word 
${\tt ababab}$ is not.  
\end{remark}

\section*{Acknowledgments}

Funded in part by a grant from NSERC.  We are grateful to the referees
for several suggestions.

\end{document}